\documentclass[envcountsame,pdftex,runningheads]{llncs}
\usepackage{pta}

\title{Quantitative Games on Probabilistic Timed Automata}
\author{Marta Kwiatkowska, Gethin Norman, and Ashutosh Trivedi}
\institute{Oxford University Computing Laboratory, Parks Road, Oxford,
  OX1 3QD} 

\begin{document}

\maketitle

%%%%%%%%%%%%%%%%%%%%%%%%%%%%%%%%%%%%%%%%%%%%%%%%%%%%%%%%%%%%%%%%%%%%%%%% 
\begin{abstract}
  Two-player zero-sum games are a well-established model for
  synthesising controllers that optimise some performance criterion.  
  In such games one player represents the controller, while the other
  describes the (adversarial) environment, and controller synthesis
  corresponds to computing the optimal strategies of the controller
  for a given criterion.  
  Asarin and Maler initiated the study of quantitative games on
  (non-probabilistic) timed automata by synthesising controllers which
  optimise the time to reach a final state.  
  The correctness and termination of their approach was dependent on
  exploiting the properties of a special class of functions, called
  \emph{simple functions},  that can be finitely represented. 
  In this paper we consider quantitative games over probabilistic
  timed automata.  
  Since the concept of simple functions is not sufficient to solve
  games in this setting, we generalise simple functions to so-called
  \emph{quasi-simple functions}. 
  Then, using this class of functions, we demonstrate that the problem
  of solving games with either expected reachability-time or expected
  discounted-time criteria on probabilistic timed automata are in
  NEXPTIME $\cap$ co-NEXPTIME.  
%  As a side result, we also show that the problem of solving
%  discounted-time games on timed automata is EXPTIME-complete. 
\end{abstract}

\section{Introduction}
Two-player zero-sum games on finite automata, as a
mechanism for supervisory controller synthesis of discrete event
systems, were introduced by Ramadge and Wonham~\cite{RW89}.
In this setting
the two players---called Min and Max---represent the `controller' and
the `environment'
and control-program synthesis corresponds to finding a winning (or
optimal) strategy of the `controller' for some given performance
objective. 
If the objectives are dependent on time, e.g.\
when the objective corresponds to completing a given set of tasks
within some deadline,  then games on timed automata are a
well-established approach for controller synthesis, see for example
\cite{AM99,ABM04,BCFL04,BC+07}. 

In this paper we extend this approach to systems which are
quantitative in terms of time \emph{and} probabilistic behaviour. 
Probabilistic information is important for modelling, e.g.,
faulty or unreliable components, the random coin flips of
distributed communication and security protocols, and performance
characteristics.  
We consider games on probabilistic timed automata
\cite{KNSS02,Jen96,Bea03}, a modelling framework for real-time systems
exhibiting both nondeterministic and probabilistic behaviour.  
We concentrate on expected
reachability-time games, where the performance objective
concerns the expected minimum time the controller can ensure
for the system to reach a target, regardless of
uncontrollable (environmental) events.
This approach has many practical applications, including job-shop
scheduling, where machines can be faulty or have variable 
execution times, and both routing and task graph scheduling
problems, where both time and stochastic behaviour is also relevant. 
We also discuss discounted-time games
where, intuitively, at each transition the system
breaks down with some non-zero probability, and the players try to
optimise the expected time to breakdown. 

\paragraph{Contributions.}
Our approach is inspired by the work of Asarin and Maler \cite{AM99}
who initiated the study of quantitative games on (non-probabilistic)
timed automata.  
Their results were dependent on exploiting the properties of a special
class of functions, called \emph{simple functions}, that can be
finitely represented.  
Since the concept of simple functions is not sufficient to solve
games in this setting, we generalise simple functions to so-called
\emph{quasi-simple functions}. 
Using this class of functions and the boundary region graph
construction \cite{JT08}, we demonstrate that the problem 
of solving games with either expected reachability-time or expected
discounted-time criteria on probabilistic timed automata are in
NEXPTIME $\cap$ co-NEXPTIME.  
%As a side result, we also show that the problem of solving
%discounted-time games on timed automata is EXPTIME-complete. 

\paragraph{Related Work.}
Hoffman and Wong-Toi~\cite{HW92} were the first to define 
and solve optimal controller synthesis problem for timed automata.
For a detailed introduction to the topic of qualitative games on timed
automata, see e.g. \cite{AMP95}. 
Asarin and Maler~\cite{AM99} initiated the study of quantitative
games on timed automata by providing a symbolic algorithm  to solve
reachability-time games. 
The works of \cite{BHPR07} and \cite{JT07} showed that the decision
version of the reachability-time game is EXPTIME-complete for timed
automata with at least two clocks.
For average-time objectives, Jurdzi{\'n}ski
and Trivedi~\cite{JT08b} showed the EXPTIME-completeness of the
problem for timed automata with two or more clocks.

A natural extension of reachability-time games for timed automata is
reachability-price games for priced timed automata.
Alur, Bernadsky, and Madhusudan~\cite{ABM04} and Bouyer et
al.~\cite{BCFL04} gave
semi-algorithms to compute the value of reachability-price
games on linearly-priced timed automata. In~\cite{BBR05} and
~\cite{BBM06} it was shown that checking the existence of optimal
strategies in a reachability-price game is undecidable for automata
with three clocks and stopwatch prices.  

We are not aware of any previous work studying games on probabilistic timed
automata.  
For a significantly different model of stochastic timed
games, \cite{BF09} show that deciding whether a
target is reachable within a given probability bound is
undecidable.   
Regarding one-player games on probabilistic timed automata, 
\cite{JKNT09} shows that a number of one-player optimisation
problems on concavely-priced probabilistic timed automata can be
reduced to solving corresponding problems on the boundary region
graph.  
We also mention \cite{KNPS06},  based on the digital
clocks approach \cite{HMP92}, which solves expected-time
(and expected-cost) reachability for a subclass of probabilistic timed
automata. 

\section{Preliminaries}
We begin by presenting the background material required in the remainder of the paper.
We assume, the sets $\Nat$ of non-negative 
integers, $\Real$ of reals
and $\Rplus$ of non-negative reals.  
For $n \in \Nat$, let $\NATS{n}$ and $\REALS{n}$ denote the sets 
$\set{0, 1, \dots, n}$, and
$\set{r \in \Real \, | \, 0 {\leq} r {\leq} n}$ respectively.
For $x {=} (x_1, \ldots, x_n) \in \Real^n$, we
define $\InfNorm{x} = \max \set{|x_i| \, | \, 1 \leq i \leq n}$.
\vspace*{-6pt}
%%%%%%%%%%%%%%%%%%%%%%%%%%%%%%%%%%%%%%%%%%%%%%%%%%%%%%%%%%%%%%%%%%%%%%%%
\subsubsection*{Probability distributions.}
A \emph{discrete probability distribution} over a countable set $Q$ is
a function $\mu: Q {\to} [0, 1]$ such that $\sum_{q \in Q} \mu(q) {=} 1$.
%Let $\DIST(Q)$ denote the set of all distributions over $Q$.
For a possible uncountable set $Q'$, we define $\DIST(Q')$ to be the set of
functions $\mu: Q' \to [0,1]$ such that the set $\supp(\mu) {=}\set{q \in Q \, | \,
  \mu(q) {>} 0}$ is countable and, over $\supp(\mu)$,
$\mu$ is a distribution. 
We say that $\mu \in \DIST(Q)$ is a \emph{point distribution}
if $\mu(q) {=} 1$ for some $q \in Q$.
\vspace*{-6pt}
%%%%%%%%%%%%%%%%%%%%%%%%%%%%%%%%%%%%%%%%%%%%%%%%%%%%%%%%%%%%%%%%%%%%%%%%
\subsubsection*{Markov decision processes.} We next introduce Markov decision processes a modelling
formalism for systems exhibiting nondeterministic and probabilistic behaviour.
\begin{definition}
  A \emph{Markov decision process} (MDP) is a tuple  $\mdp = (S, F, A, p, \pi)$
  where: 
  \begin{itemize}
  \item
    $S$ is the set of \emph{states} including a set of final states $F$; 
  \item
    $A$ is the set of \emph{actions};
  \item
    $p : S \times A \to \DIST(S)$ is a partial function called the
    \emph{probabilistic transition function};
  \item
    $\pi: S {\times} A \to \Rplus$ is the
    \emph{reward function}. 
  \end{itemize}
\end{definition}
%%%%%%%%%%%%%%%%%%%%%%%%%%%%%%%%%%%%%%%%%%%%%%%%%%%%%%%%%%%%%%%%%%%%%%%% 
We write $A(s)$ for the set of actions available at $s$, i.e., the
set of actions $a$ for which $p(s, a)$ is defined. 
%For technical convenience we assume that $A(s)$ is nonempty for all 
%$s \in S$.
In an MDP $\Mm$, if the current state is $s$, then there is a non-deterministic choice
between the actions in $A(s)$ and if action $a$ is chosen the probability of reaching
the state $s' \in S$ equals $p(s'|s, a) \rmdef p(s, a)(s')$.
\vspace*{-6pt}
%%%%%%%%%%%%%%%%%%%%%%%%%%%%%%%%%%%%%%%%%%%%%%%%%%%%%%%%%%%%%%%%%%%%%%%%
\subsubsection*{Clocks, clock valuations, regions and zones.}
We fix a constant $k \in \Nat$ and finite set of \emph{clocks} $\clocks$.
A ($k$-bounded) \emph{clock valuation} is a function 
$\nu : \clocks \to \REALS{k}$ and we write $V$ for the set of clock valuations. 
%%%%%%%%%%%%%%%%%%%%%%%%%%%%%%%%%%%%%%%%%%%%%%%%%%%%%%%%%%%%%%%%%%%%%%%% 
\begin{assum}
  Although clocks are usually allowed
  to take arbitrary non-negative values, we have restricted their values to be bounded by the constant $k$. 
  This restriction is for technical convenience and comes without
  significant loss of generality. 
\end{assum}
%%%%%%%%%%%%%%%%%%%%%%%%%%%%%%%%%%%%%%%%%%%%%%%%%%%%%%%%%%%%%%%%%%%%%%%% 
If~$\nu \in V$ and $t \in \Rplus$ then we write $\nu {+} t$ for the
clock valuation defined by $(\nu {+} t)(c) = \nu(c) {+} t$, for all 
$c \in \clocks$.
For $C \subseteq \clocks$, we write $\nu[C{:=}0]$ for the clock valuation
where $\nu[C{:=}0](c) = 0$ if 
$c \in C$, and $\nu[C{:=}0](c) = \nu(c)$ otherwise.
% A \emph{corner} is an integer clock valuation, i.e., $\alpha$ is a
% corner if $\alpha(c) \in \NATS{k}$, for every clock $c \in \clocks$.
For $X \subseteq V$, we write $\CLOS{X}$ for the smallest closed set
in~$V$ containing $X$.  
Let $X \subseteq V$ be a convex subset of clock valuations and let $F :
X \to \Real$ be a continuous function.
We write $\CLOS{F}$ for the unique continuous function 
$F' : \CLOS{X} \to \Real$, such that
$F'(\nu) = F(\nu)$ for all $\nu \in X$. 

The set of \emph{clock constraints} over $\clocks$ is the 
set of conjunctions of \emph{simple constraints}, which are constraints
of the form $c \bowtie i$ or $c {-} c' \bowtie i$, where 
$c, c' \in \clocks$, $i \in \NATS{k}$, and 
${\bowtie} \in \{<, >, =, \leq, \geq \}$. 
For every $\nu \in V$, let $\CC(\nu)$ be the set of 
simple constraints which hold in~$\nu$.  
A \emph{clock region} is a maximal set $\region \subseteq V$, such
that  $\CC(\nu) {=} \CC(\nu')$ for all $\nu, \nu' \in \region$.
Every clock region is an equivalence class of the
indistinguishability-by-clock-constraints relation, and vice versa.  
Note that $\nu$ and~$\nu'$ are in the same clock region if and only if
the integer parts of the clocks and the
partial orders of the clocks, determined by their fractional parts,
are the same in $\nu$ and $\nu'$.  
We write $[\nu]$ for the clock region of $\nu$ and, if $\zeta {=}
[\nu]$, write $\region[C{:=}0]$ for the clock region $[\nu[C{:=}0]]$. 

A \emph{clock zone} is a convex set of clock valuations, which  
is a union of a set of clock regions. We write $\zones$ for the set of
clock zones. 
For any clock zone $W$ and clock valuation $\nu$, we use the notation
$\nu \sat W$ to denote that $[\nu] \in W$.
A set of clock valuations is a clock zone if and only if it is
definable by a clock constraint.
Observe that, for every clock zone~$W$, the set $\CLOS{W}$ is also a
clock zone.

\section{Stochastic Games on Probabilistic Timed Automata}
In this section we introduce stochastic games played on probabilistic
timed automata. 
\vspace*{-6pt}
%%%%%%%%%%%%%%%%%%%%%%%%%%%%%%%%%%%%%%%%%%%%%%%%%%%%%%%%%%%%%%%%%%%%%%%%
\subsubsection*{Probabilistic timed automata.}
\label{Sec:timed-automata}
Probabilistic timed automata are a modelling framework for
real-time systems exhibiting both nondeterministic
and probabilistic behaviour. The formalism is derived by extending classical
timed automata \cite{AD94} with discrete probability distributions over edges. 
%%%%%%%%%%%%%%%%%%%%%%%%%%%%%%%%%%%%%%%%%%%%%%%%%%%%%%%%%%%%%%%%%%%%%%%% 
\begin{definition}[PTA syntax]
  A \emph{probabilistic timed automaton} (PTA) is a tuple
  $\pta = (L, \fin , \clocks, \inv, \act, E, \delta)$ where:
\begin{itemize}
\item
  $L$ is the finite set of \emph{locations}  including the set of \emph{final locations} $\fin$;
\item
  $\clocks$ is the finite set of \emph{clocks}; 
\item
  $\inv : L \to \zones$ is the \emph{invariant condition}; 
\item
  $\act$ is the finite set of \emph{actions};
\item 
  $E : L {\times}\act \to \zones$ is the \emph{action enabledness function}; 
\item
  $\delta : (L {\times} \act) \to \DIST(2^{\clocks} {\times} L)$ is the \emph{transition probability function}.
\end{itemize}
\end{definition}
A \emph{timed automaton} is a PTA with the property that  $\delta(\ell, a)$ is a
point distribution for all
$\ell \in L$ and $a \in \act$. 
When we consider a PTA as an input of an 
algorithm, its size should be understood as the sum of the sizes of
encodings of $L$, $\clocks$, $\inv$, $\act$, $E$, and $\delta$.
As usual~\cite{JLS08}, we assume that probabilities are expressed as
ratios of two natural numbers, each written in binary. 
In addition, we assume the following standard  restriction on PTAs which ensures
\emph{time divergent} behaviour. 
\begin{assum}
  We restrict attention to \emph{structurally
    non-Zeno} PTAs~\cite{Tri99,JLS08}.
\end{assum}
A \emph{configuration} of a PTA $\pta$ is a
pair $(\ell, \nu)$, where $\ell \in L$ is a location and $\nu \in V$
is a clock valuation over $\clocks$ such that 
$\nu \sat \inv(\ell)$. For any $t \in \Real$, we let $(\ell,\nu){+}t$
equal the configuration $(\ell,\nu{+}t)$. 
Informally, the behaviour of a PTA is as
follows. 
In configuration $(\ell,\nu)$ time passes before an available
action is triggered, after which a discrete probabilistic transition
occurs.  
Time passage is available only if the invariant condition $\inv(\ell)$
is satisfied while time elapses, and an action $a$ can be chosen after
time $t$ elapses only if it is enabled after time elapse, i.e.,\ 
if $\nu{+}t \sat E(\ell,a)$. 
Both the time and the action chosen are nondeterministic. 
If the action $a$ is chosen, then the probability of moving to the
location $\ell'$ and resetting all of the clocks in $C$ to 0 is given
by $\delta[\ell,a](C, \ell')$. 

Formally, the semantics of a PTA is given by 
an MDP which has both an infinite number of states
and an infinite number of transitions. 
%%%%%%%%%%%%%%%%%%%%%%%%%%%%%%%%%%%%%%%%%%%%%%%%%%%%%%%%%%%%%%%%%%%%%%%%
\begin{definition}[PTA semantics]
  Let $\pta = (L, \fin , \clocks, \inv, \act, E, \delta)$ be a
  PTA. 
  The semantics of $\pta$ is the MDP $\sem{\pta} = (S, F, A, p, \pi)$
  where 
  \begin{itemize}
  \item 
    $S \subseteq L {\times} V$, the set of states, 
    is such that $(\ell,\nu) \in S$ if and only if $\nu \sat \inv(\ell)$;
  \item
    $F = S \cap (\fin \times V)$ is the set of final states;
  \item 
    $A = \Rplus {\times} \act$ is the set of \emph{timed actions}; 
  \item
    $p : S \times A \to \DIST(S)$ is the probabilistic transition function
    such that for $(\ell,\nu) \in S$ and $(t,a) \in A$, we
    have $p( (\ell,\nu)  , (t,a) ) = \mu$ if and only if
    \begin{itemize}
    \item
      $\nu {+} t' \sat \inv(\ell)$ for all $t' \in [0, t]$;
    \item
      $\nu {+} t \sat E(\ell,a)$;
    \item
      $\mu((\ell',\nu')) =\sum_{C \subseteq \clocks \wedge
        (\nu{+}t)[C{:=0}]=\nu'} \delta[\ell,a](C,l')$  for all
      $(\ell',\nu') \in S$. 
    \end{itemize}
  \item 
    $\pi : S{\times} A {\to} \Real$ is the reward function where $\pi(s,(t, a)) {=} t$
    for $s \in S$ and $(t, a) \in A$.
  \end{itemize}
\end{definition}
In the rest of the paper, for the sake of notational convenience, we
often write $\mu(\ell, \nu)$ for $\mu((\ell, \nu))$.
% Notice that if $\pta$ is a timed automaton then $\sem{\pta}$ is a weighted finite graph. 
\vspace*{-6pt}
%%%%%%%%%%%%%%%%%%%%%%%%%%%%%%%%%%%%%%%%%%%%%%%%%%%%%%%%%%%%%%%%%%%%%%%%
\subsubsection*{Probabilistic Timed Game Arena.}
We are now in a position to introduce probabilistic timed game arenas.
\begin{definition}
  A \emph{probabilistic timed game arena} is a triplet 
  $\Tt = (\pta, L_\mMIN, L_\mMAX)$ where 
  $\pta = (L, \fin , \clocks, \inv, \act,  E, \delta)$ is a
  PTA and $(L_\mMIN, L_\mMAX)$ is a
  partition of $L$. 
\end{definition}
The semantics of a probabilistic timed game arena $\Tt$ is the
stochastic game arena $\sem{\Tt}=(\sem{\pta}, S_\mMIN, S_\mMAX)$ where
$\sem{\pta} = (S, A, F, p, \pi)$ is the semantics of $\pta$, and 
$S_\mMIN = S \cap  (L_\mMIN \times V)$ and 
$S_\mMAX = S \setminus S_\mMIN$. 
Intuitively $S_\mMIN$ is the set of states controlled by player Min,
and $S_\mMAX$ is the set of states controlled by player Max.

In a turn-based game on $\Tt$ players Min and Max
move a token along the states of the PTA in the following manner. 
If the current state is $s$, then the
player controlling the state chooses an action $(t, a) \in A(s)$ after
which state $s' \in S$ is reached with probability  
$p(s'|s, a)$.
In the next turn the player controlling the state $s'$ chooses an
action in $A(s')$ and a probabilistic transition is made accordingly. 

We say that $(s, (t, a), s')$ is a transition in $\Tt$ if 
$p(s' | s, (t, a)) {>} 0$ and a \emph{play} of $\Tt$ is a sequence  
$\seq{s_0, (t_1, a_1), s_1, \ldots} \in S {\times} (A {\times} S)^*$
such that $(s_i, (t_{i+1}, a_{i+1}), s_{i+1})$ is a transition for all
$i {\geq} 0$. 
We write $\RUNS$ ($\FRUNS$) for the sets of infinite (finite) 
plays and $\RUNS_s$ ($\FRUNS_s$) for the sets of infinite (finite)
plays starting from state~$s$.  
For a finite play $r$ let
$\LAST(r)$ denote the last state of the play. 
Let $X_i$ and $Y_i$ denote the random variables
corresponding to $i^{\text{th}}$ state and action of a play. 

A \emph{strategy} of player Min in $\Tt$ is a partial function 
$\mu: \FRUNS \to \DIST(A)$, defined for $r \in \FRUNS$ 
if and only if $\LAST(r) \in S_\mMIN$, such that
$\supp(\mu(r)) \subseteq A(\LAST(r))$.  
Strategies of player Max are defined analogously.
We write $\Sigma_\mMIN$ and $\Sigma_\mMAX$ for the set of
strategies of players Min and Max, respectively.
Let $\RUNS^{\mu, \chi}_s$ denote the subset of $\RUNS_s$
which corresponds to the set of plays in which
players play according to $\mu \in \Sigma_\mMIN$ and $\chi \in
\Sigma_\mMAX$, respectively. 
A strategy $\sigma$ is \emph{pure} if $\sigma(r)$ is a 
point distribution for all  $r \in \FRUNS$ for which it is defined,
while it is \emph{stationary} if 
$\LAST(r) {=} \LAST(r')$ implies $\sigma(r) {=} \sigma(r')$ for all  
$r, r' \in \FRUNS$. 
%We write $\Pi_\mMIN$ and $\Pi_\mMAX$ for the set of positional
%strategies of player Min and player Max, respectively.

To analyse the behaviour of a stochastic game on $\Tt$ under a
strategy pair $(\mu,\chi)$, for
every state $s$ of $\Tt$, we define a probability space 
$(\RUNS^{\mu, \chi}_s, \mathcal{F}_{\RUNS^{\mu, \chi}_s},
\PROB^{\mu, \chi}_s)$ over the set of infinite plays under
strategies $\mu$ and $\chi$ with $s$ as the initial state. 
Given a \emph{real-valued random variable} $f : \RUNS \to \Real$, we
can then define the expectation of this variable 
$\eE^{\mu, \chi}_s\set{f}$ with respect to strategy pair $(\mu,
\chi)$ when starting in $s$.  

For technical convenience we make the following standard~\cite{NS04}
assumption (a similar assumption is required for optimal expected
reachability price problem for finite MDP~\cite{dA99}):  
\begin{assum}
  \label{assum:proper-strategy}
  For every strategy pair $\mu \in \Sigma_\mMIN$, 
  $\chi \in \Sigma_\mMAX$, and state $s \in S$ we have that 
  $\lim_{i \to \infty} \PROB^{\mu, \chi}_s(X_i \in F) = 1$. 
\end{assum}
\vspace*{-6pt}
%%%%%%%%%%%%%%%%%%%%%%%%%%%%%%%%%%%%%%%%%%%%%%%%%%%%%%%%%%%%%%%%%%%%%%%%
\subsubsection*{Expected Reachability-Time Game.} In an expected reachability-time
game on $\Tt = (\pta, L_\mMIN, L_\mMAX)$ player Min attempts to reach the
final states as quickly as possible, while the objective of player Max
is the opposite.
More precisely, Min is interested in minimising her losses, while player Max is
interested in maximising his winnings where, if player Min uses the
strategy $\mu \in \Sigma_\mMIN$ and player Max uses the strategy 
$\chi \in \Sigma_\mMAX$, player Min loses the following amount to player Max:
\[
\EREACHPRICE(s, \mu, \chi) \rmdef \eE^{\mu, \chi}_s
\left\{\mbox{$\sum_{i=1}^{\min \set{i \, | \, X_i \in F}}$} 
  \pi(X_{i-1}, Y_{i}) \right\}.
\]
Observe that player Max can choose his actions to win at least an
amount arbitrarily close to 
$\sup_{\chi \in \Sigma_{\mMAX}} \inf_{\mu \in \Sigma_{\mMIN}}
\EREACHPRICE(s, \mu, \chi)$.
This is called the \emph{lower value} $\LVAL(s)$
of the expected reachability-time game starting at $s$:
\[ \begin{array}{c}
\LVAL(s)  \rmdef \sup_{\chi \in \Sigma_\mMAX} \inf_{\mu \in
  \Sigma_\mMIN} \EREACHPRICE(s, \mu, \chi) \, .
\end{array} \]
Similarly, player Min can choose to lose at most an amount
arbitrarily close to 
$\inf_{\mu \in \Sigma_{\mMIN}} \sup_{\chi \in \Sigma_{\mMAX}} 
\EREACHPRICE(s, \mu, \chi)$.   
This is called the \emph{upper value} $\UVAL(s)$ of
the game: 
\[ \begin{array}{c}
\UVAL(s) \rmdef \inf_{\mu \in \Sigma_\mMIN} \sup_{\chi \in
  \Sigma_\mMAX} \EREACHPRICE(s, \mu, \chi) \, . 
\end{array} \]
It is straightforward to verify that 
$\LVAL(s) \leq \UVAL(s)$ for all $s \in S$.  
We say that the expected reachability-time game is determined if 
$\LVAL(s) = \UVAL(s)$ for all $s \in S$. 
In this case we also say that the value of the game exists and denote 
it by $\VAL(s) = \LVAL(s) = \UVAL(s)$ for all $s \in S$. 
The results of this paper present a  proof of the following proposition.
\begin{proposition}
Expected reachability-time games are determined.
\end{proposition}
For $\mu \in \Sigma_\mMIN$ and $\chi \in \Sigma_\mMAX$ we 
define $\VAL^\mu(s) = \sup_{\chi \in \Sigma_\mMAX} \EREACHPRICE(s,
\mu, \chi)$
and 
$\VAL_\chi(s) = \inf_{\mu \in \Sigma_\mMIN} \EREACHPRICE(s, \mu, \chi)$.
For an $\varepsilon {>} 0$, we say that $\mu \in \Sigma_\mMIN$ or
$\chi \in \Sigma_\mMAX$ is \emph{$\varepsilon$-optimal} if
$\VAL^\mu(s) {\leq} \VAL(s) {+} \varepsilon$  or
$\VAL_\chi(s) {\geq} \VAL(s) {-} \varepsilon$, respectively, for all
$s \in S$.  
If an expected reachability-time game is determined, then for
every $\varepsilon {>} 0$, both players have $\varepsilon$-optimal
strategies. 
\vspace*{-6pt}
%%%%%%%%%%%%%%%%%%%%%%%%%%%%%%%%%%%%%%%%%%%%%%%%%%%%%%%%%%%%%%%%%%%%%%%%
\subsubsection*{Optimality Equations.}
We now review optimality equations for characterising the value in an
expected reachability-time game.
Let $\Tt$ be a probabilistic timed game arena and let 
$P : S \to \Rplus$.
We say that $P$ is a solution of optimality equations $\Opt(\Tt)$, and
we write $P \models \Opt(\Tt)$ if, for all $s \in S$:
\[
  P(s) = \left\{
  \begin{array}{cl}
    0 & \text{if $s \in F$}\\
    \inf_{(t, a) \in A(s)} \{ t + \sum_{s' \in S} p(s'| s, (t, a))  \cdot P(s') \}
    & \text{if $s \in S_\mMIN{\setminus} F$}\\
    \sup_{(t, a) \in A(s)} \{ t + \sum_{s' \in S} p(s'| s, (t, a))  \cdot P(s') \} 
    & \text{if $s \in S_\mMAX{\setminus} F$}
  \end{array} \right.
\]
%%%%%%%%%%%%%%%%%%%%%%%%%%%%%%%%%%%%%%%%%%%%%%%%%%%%%%%%%%%%%%%%%%%%%%%% 
Under Assumption~\ref{assum:proper-strategy},
the proof of the following proposition is routine and for details see,
for example, \cite{JT07}. 
\begin{proposition}
  \label{proposition:opt-strategies-from-opt-eqn}
  If $P \models \Opt(\Tt)$, then 
  $\VAL(s) = P(s)$ for all $s \in S$ and, for every $\varepsilon {>} 0$, both 
  players have \emph{pure} $\varepsilon$-optimal strategies.   
\end{proposition}
Using Proposition~\ref{proposition:opt-strategies-from-opt-eqn}, it follows that the problem of
solving an expected reachability-time game on $\Tt$ can be reduced to
solving the optimality equations $\Opt(\Tt)$.
In the non-probabilistic setting, Jurdzi{\'n}ski and
Trivedi~\cite{JT07}  showed that solving optimality equations for a
reachability-time game on a (non-probabilistic) timed automaton $\Tt$
can be reduced to solving a reachability-price game on an abstraction,
called the boundary region graph.
Recently~\cite{JKNT09}, we extended this result
reducing a number of
one-player optimisation problems on probabilistic timed automata to
solving corresponding problems on boundary region graphs.
In the next section, we review boundary region graph abstraction for
probabilistic timed automata and, in Section~\ref{sec:correctness}, we
argue that boundary region graph abstraction for probabilistic timed
automata is sufficient to solve expected reachability-time games.
In Section~\ref{sec:discounted-games}, we explain that expected
discounted-time games on probabilistic timed automata can also be
reduced to solving discounted-price games on their boundary region
graph. 
In Section~\ref{sec:complexity}, we discuss some implications of these
reductions on the complexity of the decision problems related to these
games.

%%%%%%%%%%%%%%%%%%%%%%%%%%%%%%%%%%%%%%%%%%%%%%%%%%%%%%%%%%%%%%%%%%%%%%%% 
\section{The Boundary Region Graph Abstraction} 
\label{sec:region-graph-construction}
In this section we review the boundary region graph for PTAs introduced in \cite{JKNT09}.
\vspace*{-6pt}
%%%%%%%%%%%%%%%%%%%%%%%%%%%%%%%%%%%%%%%%%%%%%%%%%%%%%%%%%%%%%%%%%%%%%%%%
\subsubsection*{Regions.}
A~\emph{region} is a pair $(\ell, \region)$, where $\ell$ is a
location and $\region$ is a clock region such that 
$\region \subseteq \inv(\ell)$.
For any $s{=} (\ell, \nu)$, we write $[s]$ for the
region $(\ell, [\nu])$ and $\Rr$ for the set of regions.
A set $Z \subseteq L {\times}V$ is a \emph{zone} if, for every 
$\ell \in L$,
there is a clock zone $W_\ell$ (possibly empty), such that  
$Z = \set{(\ell, \nu) \, | \, \ell \in L \wedge \nu \sat W_\ell}$.  
For a region $R {=} (\ell, \region) \in \Rr$, we write $\CLOS{R}$ for
the zone $\set{(\ell, \nu) \, | \, \nu \in \CLOS{\region}}$, recall
$\CLOS{\region}$ is the smallest closed set in~$V$ containing 
$\region$. 

For $R, R' \in \Rr$, we say that $R'$ is in the future of $R$, 
or that $R$ is in the past of $R'$, if there is $s \in R$,
$s' \in R'$ and $t \in \Rplus$ such that $s' = s{+}t$; 
we then write $R \rightarrow_* R'$.
We say that $R'$ is the \emph{time successor} of $R$ if 
$R \rightarrow_* R'$, $R {\neq} R'$, and 
$R \rightarrow_* R'' \rightarrow_* R'$ implies 
$R'' {=} R$ or $R'' {=} R'$ and write $R \rightarrow_{+1} R'$ and
$R' \leftarrow_{+1} R$.

We say that a region $R \in \Rr$ is \emph{thin} if $[s] \not= [s{+}\varepsilon]$ for every $s \in R$ 
and $\varepsilon {>} 0$;
other regions are called \emph{thick}.
We write $\Rr_\THIN$ and $\Rr_\THICK$ for the sets of thin and thick
regions, respectively. 
Note that if $R \in \Rr_\THICK$ then, for every $s \in R$, there is an
$\varepsilon > 0$, such that $[s] = [s {+} \varepsilon]$.  
Observe that the time successor of a thin region is thick, and 
vice versa. 

We say $(\ell,\nu) \in L {\times} V$ is in the \emph{closure of
  the region} $(\ell,\region)$, and we write 
$(\ell,\nu) \in \CLOS{(\ell,\region)}$, if $\nu \in \CLOS{\region}$. 
For any $\nu \in V$, $b \in \NATS{k}$ and $c \in \clocks$
such that $\nu(c) {\leq} b$, we let $\mathit{time}(\nu, (b,c)) \rmdef b {-} \nu(c)$.
Intuitively, $\mathit{time}(\nu, (b, c))$ returns the amount of
time that must elapse in $\nu$ before the clock $c$
reaches the integer value $b$. 
Note that, for any $(\ell,\nu) \in L {\times} V$ and $a \in \act$,
if $t=\mathit{time}(\nu,(b,c))$ is defined, then 
$(\ell,[\nu{+}t]) \in \Rr_\THIN$ and $\supp(p_{\pta}(\cdot \,|\, (\ell,\nu),(t,a))) \subseteq
\Rr_\THIN$. 
Observe that, for every $R' \in \Rr_\THIN$, there is a
number $b \in \NATS{k}$ and a clock $c \in \clocks$, such that, for every  
$R \in \Rr$ in the past of $R'$, we have $s \in R$ implies 
$(s {+} (b {-} s(c)) \in R'$;  
and we write $R \rightarrow_{b, c} R'$.
\vspace*{-6pt}
%%%%%%%%%%%%%%%%%%%%%%%%%%%%%%%%%%%%%%%%%%%%%%%%%%%%%%%%%%%%%%%%%%%%%%%%
\subsubsection*{The Boundary Region Graph.}
The boundary region graph is motivated by the following. 
Consider any $a \in \act$, $s = (\ell,\nu)$ and 
$R =(\ell, \region) \to_* R' = (\ell, \region')$
such that $s \in R$ and $R' \sat E(\ell, a)$.
\begin{itemize}
\item
  If $R' \in \Rr_\THICK$,  then there are infinitely many $t \in \Rplus$ such that 
  $s {+}t \in R'$.
 However, amongst all such $t$'s, for one of the boundaries of $\region'$, the
closer $\nu {+} t$ is to this boundary, the `better' the
timed action $(t, a)$ becomes for a player's objective.
However, since $R'$ is a thick region,
the set $\set{t \in \Rplus \, | \, s {+} t \in R'}$
is an open interval, and hence does not contain its boundary values.
Observe that the infimum equals $b_{-} {-} \nu(c_{-})$ where 
$R \to_{b_{-}, c_{-}} R_{-} \rightarrow_{+1} R'$
and the supremum equals $b_{+} {-} \nu(c_{+})$ where
$R \to_{b_{+}, c_{+}} R_{+} \leftarrow_{+1} R'$.
In the boundary region graph we include these `best'  timed
actions through
the actions $((b_{-}, c_{-}, a), R')$ and 
$((b_{+},c_{+}, a), R')$. 
\item
If $R' \in \Rr_\THIN$, then there
exists a unique $t \in \Rplus$ such that $(\ell, \nu{+}t) \in R'$.
Moreover since $R'$ is a thin region, there exists a clock $c \in C$
and a number $b \in \Nat$ such that $R \to_{b, c} R'$ and $t = b {-} \nu(c)$.
In the boundary region graph we summarise this `best' timed action from region $R$ via region
$R'$ through the action $((b, c, a), R')$.
\end{itemize}
Based on this intuition the
boundary region graph is defined as follows.
\begin{definition}\label{brg-def}
  Let $\pta = (L, \fin, \clocks, \inv, \act, E, \delta)$ be a
  PTA. 
  The \emph{boundary region graph} of $\pta$ is defined as
  the MDP $\RegB{\pta} = ( \hS, \hF , \hA, \hp, \hpi)$ where 
  \begin{itemize}
  \item
    $\hS = \{ ((\ell,\nu),(\ell,\region)) \, | \, 
    (\ell,\region) \in \Rr \wedge \nu \in \CLOS{\region} \}$
    and
    $\hF  = \{ ((\ell,\nu),(\ell,\region)) \in \hS \, | \,  \ell \in  \fin \}$; 
  \item
    the finite set of boundary actions 
    $\hA \subseteq (\NATS{k} {\times} \clocks {\times} \act) {\times}
    \Rr$ and for $R \in \Rr$ we let\footnote{Notice that
      $\hA(R) = \hA(s)$ for all $s = ((\ell, \nu), R) \in \hS$.} $\hA(R) = 
    \{\alpha \in \hA((\ell, \nu), R) \, | \, ((\ell, \nu), R) \in \hS
    \}$;
  \item 
    for any state $((\ell,\nu),(\ell,\region)) \in \hS$ and action
    $((b, c, a),(\ell,\zeta_a)) \in \hA$ we have 
    $\hp( (\ell,\nu),(\ell,\region)  , ((b, c, a),(\ell,\zeta_a)) )
    = \mu$ if and only if 
    \[
    \mu((\ell',\nu'),(\ell',\region')) =
    \mbox{$\sum_{C \subseteq \clocks \wedge 
      \nu_a[C{:=}0]=\nu' \wedge 
      \zeta_a[C{:=}0]=\region'}$} \delta[\ell,a](C,l')
    \]
    for all $((\ell',\nu'),(\ell',\region')) \in \hS$ where 
    $\nu_a = \nu {+} \mathit{time}(\nu, (b,c))$ and one of
    the following conditions holds: 
    \begin{itemize}
    \item 
      $(\ell,\region) \rightarrow_{b, c} (\ell,\zeta_a)$ 
      and $\zeta_a \sat E(\ell,a)$
    \item 
      $(\ell,\region) \rightarrow_{b, c} 
      (\ell,\region_{-}) \rightarrow_{+1} (\ell,\zeta_a)$ for
      some $(\ell,\region_{-})$ and $\zeta_a \sat E(\ell,a)$
    \item 
      $(\ell,\region) \rightarrow_{b, c} 
      (\ell,\region_{+}) \leftarrow_{+1} (\ell,\zeta_a)$ for
      some $(\ell,\region_{+})$ and $\zeta_a \sat E(\ell,a)$. 
    \end{itemize}
  \item
    $\hpi: \hS \times \hA \to \Real$ is such that for
    $((\ell,\nu),(\ell,\region)) \in \hS$ and $((b, c, a), R) \in
    \hA(((\ell,\nu),(\ell,\region)))$ we have 
    $\hpi(((\ell,\nu),(\ell,\region)), ((b, c, a), R)) = b - \nu(c)$. 
  \end{itemize}
\end{definition}
Although the boundary region graph is infinite,
for a fixed initial state we can restrict attention to a finite state
subgraph, thanks to the following observation \cite{JT08}.
\begin{lemma}
  \label{lemma:reachable-subgraph-is-finite-MDP}
  For any state of a boundary region graph, its
  reachable sub-graph is finite.
\end{lemma}

\section{Solving PTA Games on the Boundary Region Graph.}
We now show that the boundary region graph abstraction for PTAs is
sufficient to solve the expected reachability-time games.  
The partition of the locations of a probabilistic timed game arena
$\Tt = (\pta, L_\mMIN, L_\mMAX)$ gives rise to the partition
$(\hS_\mMIN, \hS_\mMAX)$ of the set of states $\hS$ of its boundary
region graph and let $\hTt = (\RegB{\pta}, \hS_\mMIN, \hS_\mMAX)$.

We begin by reviewing the optimality equations for an
expected reachability-time game on a boundary region graph $\hTt$.
Let $P : \hS \to \Rplus$.
We say that $P$ is a solution of optimality equations $\Opt(\hTt)$, and
we write $P \models \Opt(\hTt)$, if for any $s \in \hS$:
\[
P(s) = \left\{
  \begin{array}{cl}
    0 & \text{if $s \in \hF$}\\
    \min_{\alpha \in \hA(s)} \{  t(s, \alpha) + \sum_{s' \in S} p(s'| s, \alpha)  \cdot P(s') \}
    & \text{if $s \in \hS_\mMIN {\setminus} \hF$}\\
    \max_{\alpha \in \hA(s)} \{  t(s, \alpha) + \sum_{s' \in S} p(s'| s, \alpha)  \cdot P(s') \} 
    & \text{if $s \in \hS_\mMAX {\setminus} \hF$.}
  \end{array} \right.
\]
Before trying to solve $\Opt(\hTt)$ for a probabilistic timed game, let us consider the simpler case when $\Tt$ is a timed game. 
\vspace*{-6pt}
%%%%%%%%%%%%%%%%%%%%%%%%%%%%%%%%%%%%%%%%%%%%%%%%%%%%%%%%%%%%%%%%%%%%%%%% 
\subsubsection*{The non-probabilistic case.}
For a timed game $\Tt$ we define 
$\SUCC:\hS {\times} \hA \to \hS$ as follows:
\[
\SUCC(((\ell, \nu), R), ((b, c, a), (\ell, \region))) =
((\ell', (\nu {+} b {-} \nu(c))[C{:=}0]), (\ell', \region[C{:=}0])),
\] 
where $(C, \ell') \in \powC \times L$ is such that 
$\delta(\ell, a)(C, \ell') = 1$.
Now, using this function, the optimality equations $\Opt(\hTt)$ can
be rewritten as:  
\[
P(s) = \left\{
  \begin{array}{cl}
    0 & \text{if $s \in \hF$}\\
    \min_{\alpha \in \hA(s)} \{  t(s, \alpha) +  P(\SUCC(s, \alpha)) \}
    & \text{if $s \in \hS_\mMIN {\setminus} \hF$}\\
    \max_{\alpha \in \hA(s)} \{  t(s, \alpha) +  P(\SUCC(s, \alpha)) \} 
    & \text{if $s \in \hS_\mMAX {\setminus} \hF$.}
  \end{array} \right.
\]
Based on these equations \cite{AM99} introduced the following value
iteration algorithm.
\begin{algorithm}
  \label{algorithm:2-player-vi-nonp}
  {\bf Value iteration algorithm for (non-probabilistic) $\Opt(\hTt)$.} 
  \begin{enumerate}
  \item
    Set $i := 0$,
    $p_0(s) := 0$ if $s \in \hF$ and $p_0(s) := \infty$ otherwise.
  \item
    Set $p_{i+1} := \ImproveA(p_i)$.
  \item 
    If $p_{i+1} = p_i$ then return $p_i$, else set $i := i {+} 1$ and goto step 2.  
  \end{enumerate}
  where $\ImproveA:[\hS \to \Rplus] \to [\hS \to \Rplus]$ is such that
  for any $f : \hS \to \Rplus$ and $s \in \hS$: 
  \begin{equation}
    \label{e:imp}
    \ImproveA(f)(s) = \left\{
      \begin{array}{cl}
        0 & \text{if $s \in \hF$}\\
        \min_{\alpha \in \hA(s)} \{  t(s, \alpha) +  f(\SUCC(s, \alpha)) \}
        & \text{if $s \in \hS_\mMIN {\setminus} \hF$}\\
        \max_{\alpha \in \hA(s)} \{  t(s, \alpha) +  f(\SUCC(s, \alpha)) \} 
        & \text{if $s \in \hS_\mMAX {\setminus} \hF$.}
      \end{array} \right.
  \end{equation}
\end{algorithm}
The proof of correctness of this algorithm is reliant on the concept
of \emph{simple functions}, and certain closure properties of these
functions, which we now review. 
\begin{definition}[Simple Functions]
  Let $X \subseteq V$. 
  A function $F : X \to \Real$ is \emph{simple} if either:
  there is $e \in \Int$, such that 
  $F(\nu) = e$ for every $\nu \in X$; or 
  there are $e \in \Int$ and $c \in C$, such that $F(\nu) = e - \nu(c)$ for all $\nu \in X$.   
\end{definition}
We say a function $F: \hS \to \Rplus$ is regionally simple if for every
region $(\ell, \region) \in \Rr$ the function 
$F((\ell, \cdot), (\ell, \region))$ is simple.
Asarin and Maler~\cite{AM99} showed that the following properties hold for simple functions.
\begin{proposition}[Properties of simple functions]
  \label{prop-simple-prop}
  \begin{enumerate}
  \item If $F: X \to \Real$ is simple, then $\CLOS{F}:\CLOS{X} \to \Real$ is simple.
  \item If $F, F' : \hS \to \Real$ are regionally simple functions, then 
    $\min(F, F')$  and $\max(F, F')$ are also regionally
    simple~\footnote{
      For functions $F, F' : \hS \to \Real$ we define functions 
      $\max(F, F'), \min(F, F') : \hS \to \Real$ by
      $\max(F, F')(s) = \max \eset{F(s), F'(s)}$ and 
      $\min(F, F')(s) = \min \eset{F(s), F'(s)}$, for every $s \in \hS$.}. 
  \item
    If $F$ be regionally simple, then, for every region $R = (\ell, \region)$ and $\alpha \in \hA(R)$,
    the function  
    $t(((\ell, \cdot), R), \alpha) + F(\SUCC(((\ell, \cdot), R),
    \alpha))$ is simple.    
  \item Any decreasing sequence of regionally simple functions is
    finite. 
  \end{enumerate}
\end{proposition}
Using the first three closure properties of simple functions, it is
easy to see that the function $\ImproveA$ in (\ref{e:imp}) is such that,
if $f$ is regionally simple, then so is $\ImproveA(f)$.
Since the initial function $p_0$ is regionally simple, it is immediate
that, if Algorithm~\ref{algorithm:2-player-vi-nonp} terminates, it will
return a regionally simple solution of $\Opt(\hTt)$. 
Now, since the function $\ImproveA$ is monotonic, 
the sequence $\seq{p_0, p_1, p_2, \ldots}$ of intermediate
value functions in Algorithm~\ref{algorithm:2-player-vi-nonp} is a
decreasing sequence of regionally simple functions.
Proposition~\ref{prop-simple-prop}(4) then guarantees the termination
of the value iteration algorithm. 
Jurdzi{\'n}ski and Trivedi~\cite{JT07} show that if a solution of 
$\Opt(\hTt)$ is regionally simple then it gives a solution of
optimality equations for the original timed automaton. 
\vspace*{-6pt}
%%%%%%%%%%%%%%%%%%%%%%%%%%%%%%%%%%%%%%%%%%%%%%%%%%%%%%%%%%%%%%%%%%%%%%%%
\subsubsection*{The probabilistic case.}
In this section we consider extending the above approach to solve
$\Opt(\hTt)$ when $\Tt$ is a probabilistic timed automaton. 
Based on the optimality equations, we define the
value improvement function $\ImproveB:[\hS \to \Rplus] \to [\hS \to
\Rplus]$ such that for any $f : \hS \to \Rplus$ and $s \in \hS$: 
\begin{equation}
  \label{e:imp2}
 \ImproveB(f)(s) {=} \left\{
    \begin{array}{cl}
      0 & \text{if $s \in \hF$}\\
      \min_{\alpha \in \hA(s)} \left\{  t(s, \alpha) +  \sum_{s' \in S} p(s'| s, \alpha)  {\cdot} f(s') \right\}
      & \text{if $s \in \hS_\mMIN {\setminus} \hF$}\\
      \max_{\alpha \in \hA(s)} \left\{  t(s, \alpha) +  \sum_{s' \in S} p(s'| s, \alpha)  {\cdot} f(s') \right\} 
      & \text{if $s \in \hS_\mMAX {\setminus} \hF$.}
    \end{array} \right.
\end{equation}
It is straightforward to verify that a fixpoint of the function
$\ImproveB$ is the solution of $\Opt(\hTt)$.
By Assumption~\ref{assum:proper-strategy} and
Lemma~\ref{lemma:reachable-subgraph-is-finite-MDP}, it is immediate
that $\ImproveB$ is a contraction, and therefore
%, since $\Opt(\hTt)$has a unique solution, 
$\ImproveB$ can be 
used in a straightforward value iteration algorithm to approximate
$\Opt(\hTt)$. 
However, trying to extend the approach of Asarin and Maler~\cite{AM99}
fails since the intermediate functions in the value iteration
algorithms no longer remain regionally simple.
To overcome this problem, we present a generalisation of
simple functions, which we call \emph{quasi-simple} functions.

Before introducing quasi-simple functions, we require the partial order
$\unlhd \subseteq V {\times} V$, where for any valuations $\nu$ and $\nu'$
we have  $\nu \unlhd \nu'$ if and only if there exists a 
$t \in \Rplus$ such that for each clock $c \in C$ either 
$\nu'(c) {-} \nu(c) = t$ or $\nu(c) {=} \nu'(c)$, and  $\nu'(c){-}\nu(c) =t$
for at least one clock $c \in C$. 
In this case we also write $(\nu' {-} \nu) = t$.

\begin{definition}[Quasi-Simple Functions]
  Let $X \subseteq V$ be a subset of valuations.
  A function $F: X \to \Real$ is \emph{quasi-simple} if:
  \begin{itemize}
  \item \emph{(Lipschitz Continuous)} there exists $K {\geq} 0 $ such that
    $ |F(\nu) {-} F(\nu')| \leq K \cdot {\InfNorm{\nu {-} \nu'}}$ for all $\nu, \nu' \in X$;
  \item \emph{(Monotonically decreasing and nonexpansive w.r.t. $\unlhd$)}
    $\nu \unlhd \nu'$ implies $F(\nu) {\geq} F(\nu')$ and $F(\nu) {-} F(\nu') \leq \nu' {-} \nu$ for all $\nu, \nu' \in X$.  
  \end{itemize}
\end{definition}

\begin{proposition}[Quasi-simple functions generalise simple functions]
  Every simple function is also quasi-simple.
\end{proposition}
\begin{proof}
  Let $X \subseteq V$ be a subset of valuations and $F: X \to \Real$ a
  simple function. 
  If $F$ is constant then the proposition trivially
  follows. 
  Otherwise, there exists $b \in \Int$ and $c \in C$ such that
  $F(\nu) = b {-} \nu(c)$ for all $\nu \in X$.
  We need to show that $F$ is Lipschitz continuous, and 
  monotonically decreasing and nonexpansive w.r.t $\unlhd$. 
  
  \begin{enumerate}
  \item 
    To prove that $F$ is Lipschitz continuous, notice that 
    $|F(\nu) - F(\nu')| = |b - \nu(c) - b + \nu'(c)| = |\nu'(c) -
    \nu(c)| \leq \InfNorm{\nu -\nu'}$.
  \item 
    For $\nu, \nu' \in X$ such that $\nu \unlhd \nu'$, we
    have $F(\nu) = b {-} \nu(c) \geq b {-} \nu'(c) = F(\nu')$.
    From the first part of this proof, it trivially follows that
    $F(\nu) {-} F(\nu') \leq \nu {-} \nu'$.\qed
  \end{enumerate}
\end{proof}
We say a function $F: \hS \to \Rplus$ is regionally quasi-simple if for every
region $(\ell, \region) \in \Rr$ the function 
$F((\ell, \cdot), (\ell, \region))$ is quasi-simple.

\begin{lemma}[Properties of Quasi-Simple Functions]
  \label{lemma:prop-quasi-simple}
  \begin{enumerate}
  \item 
    If $F: X \to \Real$ is quasi-simple, then $\CLOS{F} : \CLOS{X} \to
    \Real$ is quasi-simple.  
  \item
    If $F, F' : \hS \to \Real$ are regionally quasi-simple functions,
    then $\max(F, F')$  and $\min(F, F')$ are also regionally
    quasi-simple. 
  \item  
    If $F$ is regionally quasi-simple, then, for any $R = (\ell, \region)$ and $\alpha \in \hA(R)$,
    the function  
    $t(((\ell, \cdot), R), \alpha) + \sum_{s' \in S} p(s' | ((\ell,
    \cdot), R), \alpha) \cdot F(s')$ is quasi-simple.    
  \item 
    The limit of a sequence of quasi-simple functions is quasi-simple.  
  \end{enumerate}
\end{lemma}
From the first three properties, it follows that the 
intermediate functions of $\ImproveB$ in (\ref{e:imp2})
are regionally quasi-simple.
In addition, the fourth property implies that its fixpoint is also
regionally quasi-simple, and hence the solution of $\Opt(\hTt)$ is 
regionally quasi-simple.

\begin{proposition}
  \label{prop:quasi-simple-brg}
  Let $\Tt$ be a probabilistic timed game. 
  If $P \models \Opt(\hTt)$, then $P$ is regionally quasi-simple. 
\end{proposition}

%%%%%%%%%%%%%%%%%%%%%%%%%%%%%%%%%%%%%%%%%%%%%%%%%%%%%%%%%%%%%%%%%%%%%%%%%%%% 
\section{Correctness of the Reduction}
\label{sec:correctness}
%%%%%%%%%%%%%%%%%%%%%%%%%%%%%%%%%%%%%%%%%%%%%%%%%%%%%%%%%%%%%%%%%%%%%%%%%%%% 
We now demonstrate the correctness of our results, showing
that the problem of expected reachability-time games on PTAs
can be reduced to expected reachability-price games over the boundary
region graph. 
For a given function $f : \hS \to \Real$, we define 
$\widetilde{f} : S \to \Real$ by 
$\widetilde{f}(\ell,\nu) = f((\ell,\nu),(\ell,[\nu]))$.
Formally we have the following result.
\begin{theorem}\label{main-thm}
  Let $\Tt$ be a probabilistic timed game.  If $P \models \Opt(\hTt)$,
  then $\widetilde{P} \models \Opt(\Tt)$. 
\end{theorem}
Before giving the proof we require the following property of
quasi-simple functions.  
\begin{lemma}
  \label{lemma:nondecreasing-quasi-simple}
  Let $s = (\ell, \nu) \in S$ and $(\ell, \region) \in \Rr$ such that
  $(\ell,[\nu]) \rightarrow_* (\ell, \region)$.
  If $F : \hS \to \Real$ is regionally quasi-simple, then
  the function 
  $F^\oplus_{s, \region, a}: I \to \Real$ defined as 
  \[
  F^\oplus_{s, \region, a}(t)
  = t  + \mbox{$\sum_{(C \!,\ell') \in \powC \times L}$}
  \delta[\ell,a](C,\ell') {\cdot} F((\ell',\nu^t_C), (\ell', \region^C))
  \] 
  is continuous and nondecreasing, where $I = \{t \in \Rplus \, | \,  \nu {+} t \in \region\}$,
  $\nu^t_C = \nu{+}t[C{:=}0]$ and $\region^C = \region[C{:=}0]$.
\end{lemma}
\begin{proof}[of Theorem~\ref{main-thm}]
  Suppose that $P \models \Opt(\hTt)$, to prove this theorem it is
  sufficient to show that for any $s {=}(\ell, \nu) \in S_\mMIN$ we have:
  \begin{equation}
    \label{eqn1} \begin{array}{c}
    \tP(s) = \inf_{(t{,} a) \in A(s)} 
    \left\{ t + \mbox{$\sum_{(C \!{,}\ell') \in \powC \times L}$}
    \delta[\ell{,}a](C{,} \ell') {\cdot} \tP(\ell'{,}(\nu{+}t)[C{:=}0]) \right\}
  \end{array} \end{equation}
and for any $s {=}(\ell, \nu) \in S_\mMAX$ we have:
\begin{equation}  \begin{array}{c}
    \tP(s) = \sup_{(t{,} a) \in A(s)} 
    \left\{ t + \mbox{$\sum_{(C \!{,}\ell') \in \powC \times L}$}
      \delta[\ell{,}a](C{,} \ell') {\cdot} \tP(\ell'{,}(\nu{+}t)[C{:=}0]) \right\} \, .
  \end{array} \end{equation}
In the remainder of the proof we restrict attention to Min states as the case
for Max states follows similarly. Therefore we fix $s {=} (\ell, \nu)
\in S_\mMIN$ for the remainder of the proof. 
For $a \in \act$, let $\Rr_\THIN^{a}$ and
$\Rr_\THICK^{a}$ denote the set of 
thin and thick regions respectively that are successors of
$[\nu]$ and are subsets of $E(\ell,a)$. 
Considering the right hand side (RHS) of (\ref{eqn1}) we have:
\begin{eqnarray}
  \text{ RHS of (\ref{eqn1})} 
  &=& \min_{a \in \act} \{T_{\THIN}(s, a), T_{\THICK}(s,
  a)\}\label{eqn2}, 
\end{eqnarray}
where $T_{\THIN}(s, a)$ ($T_{\THICK}(s, a)$) is the infimum (supremum) of the
RHS of (\ref{eqn1}) over all actions $(t, a)$ such that 
  $[\nu{+}t] \in \Rr_\THIN^{a}$ ($[\nu{+}t] \in \Rr_\THICK^{a}$). 
  For the first term we have:
  \begin{eqnarray*}
    T_{\THIN}(s{,} a) & = & \min_{(\ell{,} \region) \in \Rr_\THIN^{a}} 
      \inf_{\substack{ t \in \Real \wedge \\ \nu+t \in \region}}
      \left\{ t  + 
        \mbox{$\sum\limits_{(C \!{,}\ell') \in \powC \times L}$} 
        \hspace{-0.3cm}\delta[\ell{,}a](C{,} \ell') {\cdot}
        \tP(\ell'{,}\nu^t_C) \right\} \\
    & = & \min_{(\ell{,} \region) \in \Rr_\THIN^{a}} 
    \inf_{\substack{ t \in \Real \wedge \\ \nu+t \in \region}}
    \left\{ t  + 
      \mbox{$\sum\limits_{(C \!{,}\ell') \in \powC \times L}$} 
      \hspace{-0.3cm}\delta[\ell{,}a](C{,} \ell') {\cdot}
      P((\ell'{,}\nu^t_C){,} (\ell'{,} \region_C)) \right\}\\
    & = &     \min_{(\ell{,} \region) \in \Rr_\THIN^{a}}  
    \left\{ t^{(\ell{,}\region)} + 
      \mbox{$\sum\limits_{(C \!{,}\ell') \in \powC \times L}$} 
      \hspace{-0.3cm}\delta[\ell{,}a](C{,} \ell') {\cdot}
      P((\ell'{,}\nu^{t^{(\ell{,}\region)}}_C){,} (\ell'{,} \region^C)) \right\}
  \end{eqnarray*}
  where $\nu^t_C$ denotes the clock valuation $(\nu{+}t)[C{:=}0]$,
  $t^{(\ell,\region)}$ the time to reach the region $R$ from $s$ and
  $\region^C$ the region $\region[C{:=}0]$. 
  Considering the second term of (\ref{eqn2}) we have 
  \begin{eqnarray*}
    T_{\THICK}(s{,} a) &=& \min_{(\ell{,} \region) \in \Rr_\THICK^{a}}
      \inf_{\substack{ t \in \Real \wedge \\ \nu+t \in \region}} 
      \left\{ t + 
        \mbox{$\sum\limits_{(C \!{,}\ell') \in \powC \times L}$} 
        \hspace{-0.3cm}\delta[\ell{,}a](C{,}\ell') {\cdot} 
        \tP(\ell'{,}\nu^t_C) \right\} \\
    & = & \min_{(\ell{,} \region) \in \Rr_\THICK^{a}}
    \inf_{\substack{ t \in \Real \wedge \\ \nu+t \in \region}} 
    \left\{ t + 
      \mbox{$\sum\limits_{(C \!{,}\ell') \in \powC \times L}$}
      \hspace{-0.3cm}\delta[\ell{,}a](C{,}\ell') {\cdot} 
      P((\ell'{,}\nu^t_C){,} (\ell'{,} \region^C)) \right\} \\
    & = &
    \min_{(\ell{,} \region) \in \Rr_\THICK^{a}}
    \inf_{ \substack{t^{s}_{R_{-}} {<} t {<} t^{s}_{R_{+}} \\  
        R {\leftarrow_{+1}} R_{-} \\  R {\rightarrow_{+1}} R_{+}}} 
    \left\{ t
      + \mbox{$\sum\limits_{(C \!{,}\ell') \in \powC \times L}$} 
      \hspace{-0.3cm}\delta[\ell{,}a](C{,}\ell') {\cdot}
      P((\ell'{,}\nu^t_C){,} (\ell'{,} \region^C)) \right\}
  \end{eqnarray*}
  From Proposition~\ref{prop:quasi-simple-brg} it follows that
  $P$ is regionally quasi-simple and, 
  from Lemma~\ref{lemma:nondecreasing-quasi-simple}, the function: 
  \[
  t  
  + \mbox{$\sum_{(C \!,\ell') \in \powC \times L}$} \delta[\ell,a](C,\ell') {\cdot}
  P((\ell',\nu^t_C), (\ell', \region^C))
  \]
  is continuous and nondecreasing over $\{t \, | \, \nu{+}t \in \region\}$.
  Therefore it follows that
  \begin{eqnarray*}
   T_{\THICK}(s, a) & = & \min_{(\ell{,} \region) \in \Rr_\THICK^{a}}
    \min_{\substack{t =t^{s}_{R_{-}} {,} t^{s}_{R_{+}} \\ 
        (\ell{,}\region) \leftarrow_{+1} R_{-} \\  
        (\ell{,}\region) \rightarrow_{+1} R_{+}}} 
    \left\{ t 
      + \hspace{-0.5cm} \mbox{$\sum\limits_{(C \!{,}\ell') \in \powC \times L}$} 
      \hspace{-0.3cm}\delta[\ell{,}a](C{,}\ell') {\cdot}
      P((\ell'{,}\nu^t_C){,} (\ell'{,} \region^C)) \right\}
  \end{eqnarray*}
  Substituting the values of $T_{\THIN}(s, a)$ and $T_{\THICK}(s, a)$
  into (\ref{eqn2}) and observing that, for any thin region 
  $(\ell,\region) \in \Rr_\THIN^a$, there exist $b \in \Int$ and 
  $c \in C$ such that $\nu {+} (b {-} \nu(c)) \in \region$,   
  it follows from Definition~\ref{brg-def} that RHS of (\ref{eqn1})
  equals: 
  \[ \begin{array}{c}
    \min_{\alpha \in \hA(s{,}[s])}
    \left\{ t((s{{,}} [s]){,} \alpha) {+} 
      \sum\limits_{(s'{,}R') \in \hS}
      \hp( (s'{,}R') | (s{,}[s]), \alpha) {\cdot}
      P(s'{,} R') \right\}
  \end{array} \]
  which by definition equals $\tP(s)$ as required. \qed
\end{proof}

\section{Expected Discounted-Time Games}
\label{sec:discounted-games}
Let $\Tt = (\pta, L_\mMIN, L_\mMAX)$ be a probabilistic timed game
arena and $\lambda \in [0, 1)$ be a discount factor.
In an expected discounted-time game starting from state $s$ and for a
strategy pair $\mu,\chi$, player Min loses the following amount to
player Max: 
\[
\EDISCTPRICE(s, \mu, \chi) \rmdef \eE^{\mu, \chi}_s
\left\{\mbox{$\sum_{i=1}^{\infty}$} 
  \lambda^i \cdot \pi(X_{i-1}, Y_{i}) \right\} \, .
\]
The concepts for the expected discounted-time game are defined in an
analogous manner to that of expected reachability-time games. 
A reduction from expected discounted-time game to expected
reachability-time game is standard~\cite{Shapley53}.
Therefore, using the techniques presented in this paper one can reduce
the problem of solving expected discounted-time games on $\Tt$ to
solving corresponding problem on $\hTt$. 
Similarly a (non-probabilistic) discounted-time game on timed
automata can be reduced to solving discounted-price games on
(non-probabilistic) boundary region graph.
\begin{proposition}[Discounted-Time Games]
  \label{proposition:discounted-results}
  \begin{enumerate}
  \item Expected discounted-time games on probabilistic timed
    automata can be reduced to expected discounted-price games on the
    corresponding boundary region graph.  
  \item Discounted-time games on timed automata
    can be reduced to discounted-price games on the corresponding
    (non-probabilistic) boundary region graph. 
  \end{enumerate}
\end{proposition}

\section{Complexity}
\label{sec:complexity}

%\begin{theorem}\label{discount-thm}
%  Discounted-time games are EXPTIME-complete.
%\end{theorem}
%\begin{proof}
%  Using the ideas from~\cite{ZP96} the problem of solving average-time
%  games on timed automata can be reduced, in polynomial time, to
%  solving discounted-time games on timed automata with an
%  equal number of clocks.
%  Since the problem of solving average-time games is
%  EXPTIME-complete~\cite{JT08b} on timed automata with at least
%  two clocks; discounted-time games are EXPTIME-hard on timed automata
%  with at least two clocks. 
  
%  From Proposition~\ref{proposition:discounted-results} it follows that, to solve a
%  discounted-time game on a timed automaton, it is sufficient to solve
%  the discounted-price game on the finite (non-probabilistic) boundary
%  region graph of the automaton.  
%  Observe that every region, and hence also every configuration of
%  the game, can be written down in polynomial space, and that every move
%  of the game can be simulated in polynomial time.
%  Therefore, the value of the game can be computed by a
%  straightforward alternating PSPACE algorithm, and hence the problem
%  is in EXPTIME since APSPACE = EXPTIME.
%  \qed
%\end{proof}

\begin{theorem}\label{reach-thm}
  The expected reachability-time games and the expected
  discounted-time games are EXPTIME-hard and they are in
  NEXPTIME $\cap$ co-NEXPTIME.
\end{theorem}
\begin{proof}
  The EXPTIME-hardness of expected reachability-time games and
  expected discounted-time games on probabilistic timed automata with
  two or more clocks follows from the fact that corresponding one-player games
  are EXPTIME-complete~\cite{JKNT09}. 
  The membership in NEXPTIME $\cap$ co-NEXPTIME follows from the
  reduction to the boundary region graph, and the observations that:
  size of the boundary region graph is exponential in the size of the
  PTA; and the complexity of solving expected reachability-price games
  and expected discounted-price games on finite MDP is in NP $\cap$
  co-NP. \qed 
\end{proof}

\section{Conclusion}
In this paper we have employed the boundary region graph to solve 
quantitative games over probabilistic timed automata.  
The approach is based on extending the class of simple functions
introduced in \cite{AM99} to quasi-simple functions. 
Our results demonstrate that the problem of solving games with either
expected reachability-time or expected discounted-time criteria on
PTA are in NEXPTIME $\cap$ co-NEXPTIME.
Future work includes finding practical symbolic zone-based algorithms
to solve quantitative games on timed automata and, perhaps more
ambitiously, games on PTA.  
Regarding the other quantitative games on PTA, we conjecture that it is
possible to reduce expected average-time games on PTA to mean payoff
games on boundary region graph. 
However, the techniques presented in this paper are insufficient to
demonstrate such a reduction.
\vskip8pt
\noindent {\bf Acknowledgements.}
The authors are supported in part by EPSRC grants EP/D076625 and EP/F001096.
In addition, the third author wishes to thank Vojt\v{e}ch Forejt for
some stimulating discussions on this class of games.

\bibliographystyle{plain} 
\bibliography{papers}

\begin{thebibliography}{10}

\bibitem{ABM04}
R.~Alur, M.~Bernadsky, and P.~Madhusudan.
\newblock Optimal reachability for weighted timed games.
\newblock In {\em Proc. ICALP'04}, volume 3142 of {\em LNCS}. Springer, 2004.

\bibitem{AD94}
R.~Alur and D.~Dill.
\newblock A theory of timed automata.
\newblock In {\em Theoretical Computer Science}, volume 126, 1994.

\bibitem{AM99}
E.~Asarin and O.~Maler.
\newblock As soon as possible: Time optimal control for timed automata.
\newblock In {\em Proc. HSCC'99}, volume 1569 of {\em LNCS}. Springer, 1999.

\bibitem{AMP95}
E.~Asarin, O.~Maler, and A.~Pnueli.
\newblock Symbolic controller synthesis for discrete and timed systems.
\newblock In {\em Hybrid Systems II}, volume 999 of {\em LNCS}. Springer, 1995.

\bibitem{Bea03}
D.~Beauquier.
\newblock Probabilistic timed automata.
\newblock {\em Theoretical Computer Science}, 292(1), 2003.

\bibitem{BC+07}
G.~Behrmann, A.~Cougnard, A.~David, E.~Fleury, K.~Larsen, and D.~Lime.
\newblock {UPPAAL-Tiga}: Time for playing games!
\newblock In {\em Proc. CAV'07}, volume 4590 of {\em LNCS}. Springer, 2007.

\bibitem{BBM06}
P.~Bouyer, T.~Brihaye, and N.~Markey.
\newblock Improved undecidability results on weighted timed automata.
\newblock In {\em Information Processing Letters}, volume~98. Elsevier, 2006.

\bibitem{BCFL04}
P.~Bouyer, F.~Cassez, E.~Fleury, and K.~G. Larsen.
\newblock Optimal strategies in priced timed game automata.
\newblock In {\em {FSTTCS}'04}, volume 3328 of {\em LNCS}. Springer, 2004.

\bibitem{BF09}
P.~Bouyer and V.~Forejt.
\newblock Reachability in stochastic timed games.
\newblock In {\em Proc. ICALP'09}, volume 5556 of {\em LNCS}. Springer, 2009.

\bibitem{BHPR07}
T.~Brihaye, T.~A. Henzinger, V.~S. Prabhu, and J.~Raskin.
\newblock Minimum-time reachability in timed games.
\newblock In {\em Proc. ICALP'07}, volume 4596 of {\em LNCS}. Springer, 2007.

\bibitem{BBR05}
Thomas Brihaye, V{\'e}ronique Bruy{\`e}re, and Jean-Fran\c{c}ois Raskin.
\newblock On optimal timed strategies.
\newblock In {\em Proc. FORMATS'05}, volume 3829 of {\em LNCS}. Springer, 2005.

\bibitem{dA99}
L.~de~Alfaro.
\newblock Computing minimum and maximum reachability times in probabilistic
  systems.
\newblock In {\em Proc. CONCUR'99}, volume 1664 of {\em LNCS}. Springer, 1999.

\bibitem{HMP92}
T.~A. Henzinger, Z.~Manna, and A.~Pnueli.
\newblock What good are digital clocks?
\newblock In {\em Proc. ICALP'92}, volume 623 of {\em LNCS}. Springer, 1992.

\bibitem{HW92}
G.~Hoffmann and H~Wong-Toi.
\newblock The input-output control of real-time discrete event systems.
\newblock In {\em In the Proceedings of 13th IEEE Real-Time Systems Symposium},
  1992.

\bibitem{Jen96}
H.~Jensen.
\newblock Model checking probabilistic real time systems.
\newblock In {\em Proc. 7th Nordic Workshop on Programming Theory}, Report
  86:247--261. Chalmers University of Technology, 1996.

\bibitem{JKNT09}
M.~Jurdzi{\'n}ski, M.~Kwiatkowska, G.~Norman, and A.~Trivedi.
\newblock Concavely-priced probabilistic timed automata.
\newblock In {\em Proc. CONCUR'09}, volume 5710 of {\em LNCS}, pages
  415–--430. Springer, 2009.

\bibitem{JLS08}
M.~Jurdzi{\'n}ski, J.~Sproston, and F.~Laroussinie.
\newblock Model checking probabilistic timed automata with one or two clocks.
\newblock {\em Logical Methods in Computer Science}, 4(3), 2008.

\bibitem{JT07}
M.~Jurdzi{\'n}ski and A.~Trivedi.
\newblock Reachability-time games on timed automata.
\newblock In {\em Proc. ICALP'07}, volume 4596 of {\em LNCS}. Springer, 2007.

\bibitem{JT08b}
M.~Jurdzi{\'n}ski and A.~Trivedi.
\newblock Average-time games.
\newblock In {\em Proc. FSTTCS'08}, volume~2 of {\em Leibniz International
  Proceedings in Informatics}. Schloss Dagstuhl, 2008.

\bibitem{JT08}
M.~Jurdzi{\'n}ski and A.~Trivedi.
\newblock Concavely-priced timed automata.
\newblock In {\em Proc. FORMATS'08}, volume 5215 of {\em LNCS}. Springer, 2008.

\bibitem{KNPS06}
M.~Kwiatkowska, G.~Norman, D.~Parker, and J.~Sproston.
\newblock Performance analysis of probabilistic timed automata using digital
  clocks.
\newblock {\em Formal Methods in System Design}, 29, 2006.

\bibitem{KNSS02}
M.~Kwiatkowska, G.~Norman, R.~Segala, and J.~Sproston.
\newblock Automatic verification of real-time systems with discrete probability
  distributions.
\newblock {\em Theoretical Computer Science}, 282, 2002.

\bibitem{NS04}
A.~Neyman and S.~Sorin, editors.
\newblock {\em Stochastic Games and Applications}, volume 570 of {\em NATO
  Science Series C}.
\newblock Kluwer Academic Publishers, 2004.

\bibitem{RW89}
P.~J. Ramadge and W.~M. Wonham.
\newblock The control of discrete event systems.
\newblock In {\em In Proceedings of the IEEE}, volume 77(1), 1989.

\bibitem{Shapley53}
L.~S. Shapley.
\newblock Stochastic games.
\newblock {\em Proc. Nat. Acad. Sci. U.S.A.}, 39, 1953.

\bibitem{Tri99}
S.~Tripakis.
\newblock Verifying progress in timed systems.
\newblock In {\em Proc. ARTS'99}, volume 1601 of {\em LNCS}. Springer, 1999.

\bibitem{ZP96}
U.~Zwick and M.~Paterson.
\newblock The complexity of mean payoff games on graphs.
\newblock {\em Theoretical Computer Science}, 158, 1996.

\end{thebibliography}

\newpage
\appendix
\section{Proof of Proposition~\ref{proposition:opt-strategies-from-opt-eqn}}
\begin{proof}[of Proposition~\ref{proposition:opt-strategies-from-opt-eqn}]
  We show that for every $\varepsilon {>} 0$, there exists a pure
  strategy $\mu_\varepsilon : \FRUNS \to A$ for player Min, such that
  for every strategy $\chi$ for player Max, we have 
  $\EREACHPRICE(s, \mu_\varepsilon, \chi) \leq P(s) {+} \varepsilon$. 
  The proof, that for every $\varepsilon {>} 0$, there exists a 
  pure strategy $\chi_\varepsilon : \FRUNS \to A$
  for player Max, such that for every strategy $\mu$ for player Min,
  we have $\EREACHPRICE(s, \mu, \chi_\varepsilon)) \geq P(s) -
  \varepsilon$, follows similarly. 
  Together, these facts imply that $P$ is equal to the value function
  of the expected reachability-time game, and the pure strategies
  $\mu_\varepsilon$ and $\chi_\varepsilon$, defined in the proof below
  for all $\varepsilon {>} 0$, are $\varepsilon$-optimal. 
  
  Let us fix $\varepsilon {>} 0$ and $\mu_\varepsilon$ be a pure strategy where
  for any $n \in \Nat$ and finite play $r \in \FRUNS$ of length $n$,
  $\mu_\varepsilon(r) = (t, a)$ is such that
  \[ \begin{array}{c}
  t +  \sum_{s' \in S} p(s'| \LAST(r), (t, a)) \cdot P(s')
  \leq P(\LAST(r)) {+} \frac{\varepsilon}{2^{n+1}} \, .
  \end{array} \] 
  Observe that for every state $s \in S_\mMIN$ and for every
  $\varepsilon' > 0$, there is a $\varepsilon'$-optimal timed action
  because $P \models \Opt(\Tt)$.  
  
  Again using the fact that $P \models \Opt(\Tt)$, it follows that,
  that for any $s \in S_\mMAX \setminus F$ and $(t, a) \in A$, we have 
  \begin{equation}\label{equation:max-T}
   \begin{array}{c}
 P(s)    \geq  t {+} \sum_{s' \in S} p(s'| s, a) \cdot P(s') \, .
\end{array}  \end{equation}
Now for an arbitrary strategy $\chi$ for player Max,
it follows by induction that for any $n \geq 1$:
\begin{equation}
    \label{e:reachprice} 
\begin{array}{rcl}
    P(s) & \geq &  \eE^{\mu_\varepsilon, \chi}_s  
    \left\{ \sum_{i=1}^{\min\set{i \, | \, X_i \in F}}
      \pi(X_{i-1}, Y_i) \right\} \\
    && + \sum_{s' \in S \setminus F} 
    \PROB^{\mu_\varepsilon,\chi}_s(X_n {=} s')\cdot P(s') 
    - (1 {-} \frac{1}{2^n}) {\cdot} \varepsilon \, .  
  \end{array}
\end{equation}
  Using Assumption~\ref{assum:proper-strategy}, we have
  $\lim_{n\to\infty} \mbox{$\sum_{s' \in S \setminus F}$}
  \PROB^{\mu_\varepsilon, \chi}_s(X_n {=} s') = 0$, and therefore taking the limit in (\ref{e:reachprice}) we get the inequality:
  \[
  P(s) \geq \eE^{\mu,\chi}_s\set{ 
    \mbox{$\sum_{i=1}^{\min\{i \, | \, X_i \in F \}}$} 
    \pi(X_{i-1}, Y_i)}  - \varepsilon 
  = \EREACHPRICE(s, \mu_\varepsilon, \chi)  - \varepsilon.
  \]
  which completes the proof.  \qed
\end{proof}

\section{Proof of Lemma~\ref{lemma:prop-quasi-simple}}
The proof of Lemma~\ref{lemma:prop-quasi-simple} follows from Propositions~\ref{p1}-\ref{p4} below.
Note that since every quasi-simple function $F : X \to \Real$ is Lipschitz
continuous, and hence Cauchy continuous, it can be uniquely extended
to closure of its domain $X$. 
The properties of quasi-simple function are trivially met by such
extensions. 

\begin{proposition}
  \label{p1}
  If $F: X \to \Real$ is quasi-simple, then $\CLOS{F} : \CLOS{X} \to
  \Real$ is quasi-simple.  
\end{proposition}

\begin{proposition}
  \label{p2}
  If $F, F' : \hS \to \Real$ are regionally quasi-simple functions,
  then $\max(F, F')$  and $\min(F, F')$ are also regionally
  quasi-simple. 
\end{proposition}
\begin{proof}
  To prove this proposition, it is sufficient to show that pointwise
  minimum and maximum of quasi-simple functions are quasi-simple. 
  Let $f, f': X \subseteq V \to \Real$ be quasi-simple. 
  We need to show that $\max(f, f')$ and $\min(f, f')$ are
  quasi-simple.
  
  Notice that $\max(f, f')$ and $\min(f, f')$ are Lipschitz
  continuous, as pointwise minimum and maximum of a finite set of
  Lipschitz continuous functions is Lipschitz continuous.
  
It therefore remains to show that $\max(f, f')$ and $\min(f, f')$ are monotonically
  decreasing and nonexpansive w.r.t $\unlhd$.
  Consider any $\nu,\nu' \in X$ such that $\nu_1 \unlhd \nu_2$.
  Since $f$ and $f'$ are quasi-simple, by definition $f$ and $f'$ are monotonically
  decreasing, and hence
  $f(\nu_1) \geq f(\nu_2)$ and $f'(\nu_1) \geq f'(\nu_2)$. 
Now since
  \[
  \max(f, f')(\nu_1) = \max \Set{f(\nu_1), f'(\nu_1)} \geq 
  \max \Set{f(\nu_2), f'(\nu_2)} = \max(f, f')(\nu_2),
  \]
  it follows that $\max(f, f')$ is monotonically decreasing w.r.t 
  $\unlhd$. 
  In an analogous manner we show that $\min(f, f')$ is monotonically
  decreasing w.r.t 
  $\unlhd$. 

  Again since $f$ and $f'$ are quasi-simple, we have that they are nonexpansive, i.e.,
  $f(\nu_1) {-} f(\nu_2) \leq \nu_2 {-} \nu_1$ and 
  $f'(\nu_1) {-} f'(\nu_2) \leq \nu_2 {-} \nu_1$.
  To show $\max(f, f')$ is nonexpansive, there are the following four cases to consider.
  \begin{enumerate}
  \item If $f(\nu_1) \geq f'(\nu_1)$ and $f(\nu_2) \geq f'(\nu_2)$, then
 $\max(f, f')(\nu_1) {-} \max(f,f')(\nu_2) = f(\nu_1) {-}
    f(\nu_2) \leq \nu_2 {-} \nu_1$.
  \item If $f'(\nu_1) \geq f(\nu_1)$ and $f'(\nu_2) \geq f(\nu_2)$, then
    $\max(f, f')(\nu_1) {-} \max(f,f')(\nu_2) = f'(\nu_1) {-}
    f'(\nu_2) \leq \nu_2 {-} \nu_1$.
  \item If $f(\nu_1) \geq f'(\nu_1)$ and $f'(\nu_2) \geq f(\nu_2)$, then
    $\max(f, f')(\nu_1) {-} \max(f,f')(\nu_2) = f(\nu_1) {-} f'(\nu_2) 
    \leq f(\nu_1) {-} f(\nu_2) \leq \nu_2 {-} \nu_1$.    
  \item If $f'(\nu_1) \geq f(\nu_1)$ and $f(\nu_2) \geq f'(\nu_2)$, then
    $\max(f, f')(\nu_1) {-} \max(f,f')(\nu_2) = f'(\nu_1) {-} f(\nu_2) 
    \leq f'(\nu_1) {-} f'(\nu_2) \leq \nu_2 {-} \nu_1$.    
  \end{enumerate}
  Since these are all the possible cases to consider, $\max(f, f')$ is nonexpansive w.r.t $\unlhd$. 
  Similarly show $\min(f, f')$ is nonexpansive completing the proof.
  \qed
\end{proof}

\begin{proposition}
  \label{p3}
  If $F$ is regionally quasi-simple, then for any $R = (\ell,
  \region)$ and $\alpha \in \hA(R)$ the function  
  $t(((\ell, \cdot), R), \alpha) + \sum_{s' \in S} p(s' | ((\ell,
  \cdot), R), \alpha) \cdot F(s')$ is quasi-simple.    
\end{proposition}
\begin{proof}
  Let $F$ be regionally quasi-simple and fix a region $R = (\ell,\region)$ and a boundary action 
  $\alpha =  ((a, b, c), (\ell, \region')) \in \hA(R)$ .
  We need to show that the function 
\[ \begin{array}{c}
F^\oplus_{\alpha, R}(\cdot) = 
  t (((\ell, \cdot), R), \alpha) + \sum_{s' \in S} p(s' | ((\ell,
  \cdot), R), \alpha) \cdot F(s')
\end{array} \] on the domain 
  $D = \{\nu \in V \, | \,  ((\ell, \nu), R) \in \hS \}$ is
  quasi-simple. 
  Let us first simplify the function $F^\oplus_{\alpha,R}$. For any $\nu \in D$ we have:
  \begin{eqnarray*}
    F^\oplus_{\alpha, R}(\nu) &=& t (((\ell, \nu), R), \alpha) 
    + \sum_{s' \in S} p(s' | ((\ell, \nu), R), \alpha) \cdot F(s')\\
    &=& (b {-} \nu(c)) 
    + \sum_{s' \in S} p(s' | ((\ell, \nu), R), \alpha) \cdot F(s')\\ 
    &=& (b {-} \nu(c) ) +
    \mbox{$\sum_{(C \!{,}\ell') \in \powC \times L}$}
    \delta[\ell{,}a](C{,} \ell') {\cdot} 
    F((\ell'{,}\nu_{\alpha,C}), (\ell', \region'[C{:=}0])) \\
    &=& (b {-} \nu(c) ) +
    \mbox{$\sum_{(C \!{,}\ell') \in \powC \times L}$}
    \delta[\ell{,}a](C{,} \ell') {\cdot} 
    F(s_{\ell',\nu,\alpha,C} )
  \end{eqnarray*}
  where $\nu_{\alpha,C} = (\nu{+} (b {-} \nu(c)))[C{:=}0])$ and $s_{\ell',\nu,\alpha,C}  = ((\ell'{,}\nu_{\alpha,C}), (\ell', \region'[C{:=}0]))$.
  
  Next using this simplified version we demonstrate that $F^\oplus_{\alpha, R}$ is Lipschitz
  continuous.  If $F$ is Lipschitz continuous with constant $K$, then $|F^\oplus_{\alpha, R}(\nu) {-} F^\oplus_{\alpha, R}(\nu')|$
equals
  \begin{eqnarray*}
    \lefteqn{\hspace*{-2cm} |\nu'(c) {-} \nu(c)| +  \sum_{(C \!{,}\ell') \in \powC \times L}
    \delta[\ell{,}a](C{,} \ell') {\cdot} 
    \left| 
      F(s_{\ell',\nu',\alpha,C}) 
      {-}    F(s_{\ell',\nu,\alpha,C})
    \right|} \\ 
    &\leq &  |\nu'(c) {-} \nu(c)| + 
    \sum_{(C \!{,}\ell') \in \powC \times L}
    \delta[\ell{,}a](C{,} \ell') {\cdot} K \cdot \InfNorm{\nu{-}\nu'}\\
    & = & |\nu'(c) {-} \nu(c)| +  K \cdot \InfNorm{\nu{-}\nu'}
    \leq (1 + K) \cdot \InfNorm{\nu{-}\nu'}.
  \end{eqnarray*}
  The first inequality follows from the fact that $F$ is Lipschitz
  constant with constant $K$.
  Hence it follows that $F^\oplus_{\alpha, R}$ is Lipschitz constant
  with constant $(1+K)$.
  
  It therefore remains to show that $F^\oplus_{\alpha, R}$ is
  monotonically decreasing and nonexpansive w.r.t $\unlhd$.
  Consider any $\nu, \nu' \in V$ such that 
  $\nu \unlhd \nu'$ and $\nu' {-} \nu = d$.
  We have the following  two cases to consider.
\begin{itemize}
  \item
    If $\nu(c) = \nu'(c)$, then for any set 
    $(C, \ell') \in \powC \times L$ 
    we have that $(\nu{+} b {-} \nu(c))[C{:=}0] \unlhd (\nu{+} b {-}
    \nu'(c))[C{:=0}]$, and hence  $F(s_{\ell',\nu,\alpha,C}) {-}  F(s_{\ell',\nu',\alpha,C})$
is nonnegative for all $(C, \ell') \in \powC \times L$.
    Moreover, since $F$ is nonexpansive, we have that $F(s_{\ell',\nu,\alpha,C}) -  F(s_{\ell',\nu',\alpha,C}) \leq d$.
    It follows that $F^\oplus_{\alpha, R}$ is monotonically
    decreasing and non-expansive as
    \[
    F^\oplus_{\alpha,R}(\nu) {-} F^\oplus_{\alpha,R}(\nu') 
    = \mbox{$\sum_{(C \!{,}\ell') \in \powC \times L}$}
    (F(s_{\ell',\nu,\alpha,C}) -  F(s_{\ell',\nu',\alpha,C})) 
    \]  
   and  $\nu' {-} \nu = d$.
  \item
    If $\nu'(c) {-} \nu(c) = d$, then for any $(C, \ell') \subseteq \powC \times L$ we have that 
    \[
(\nu'{+} b {-} \nu'(c))[C{:=0}] \unlhd (\nu{+} b {-} \nu(c))[C{:=}0]
\]
which implies that
$F(s_{\ell',\nu,\alpha,C}) {-}  F(s_{\ell',\nu',\alpha,C})$
    is nonpositive for all $(C, \ell') \in \powC \times L$.
    Moreover since $F$ is nonexpansive, we have that $F(s_{\ell',\nu,\alpha,C}) {-}  F(s_{\ell',\nu',\alpha,C})
    \leq d$.
    Similarly to the case above we have that $F^\oplus_{\alpha,R}$ is monotonically decreasing  and nonexpansive.
  \end{itemize}
  The proof is now complete.
  \qed
\end{proof}
The following proposition is immediate as the limit of Lipschitz
continuous functions is Lipschitz continuous, and
the limit of monotonically decreasing and nonexpansive functions is
monotonically decreasing and nonexpansive.
\begin{proposition}
  \label{p4}
  The limit of a sequence of quasi-simple functions is quasi-simple.
\end{proposition}

\section{Proof of Lemma~\ref{lemma:nondecreasing-quasi-simple}}
\begin{proof}[of Lemma~\ref{lemma:nondecreasing-quasi-simple}]
  Let $s = (\ell, \nu) \in S$ and $(\ell, \region) \in \Rr$ be such that
  $(\ell,[\nu]) \rightarrow_* (\ell, \region)$ and
  let $F : \hS \to \Real$ be regionally quasi-simple.
  We wish to show that the function 
  $F^\oplus_{s, \region, a}: I \to \Real$ defined as 
  \[
  F^\oplus_{s, \region, a}(t)
  \rmdef t  + \mbox{$\sum_{(C \!,\ell') \in \powC \times L}$}
  \delta[\ell,a](C,\ell') {\cdot} F((\ell',\nu^t_C), (\ell', \region^C))
  \] 
  is continuous and nondecreasing, where $I = \{t \in \Rplus \, | \, \nu {+} t \in \region \}$,
  $\nu^t_C = \nu+t[C{:=}0]$ and $\region^C = \region[C{:=}0]$.
  .
  
  Let $t_1, t_2 \in I$ are such that $t_1 \leq t_2$.
  To prove this proposition we need to show that 
  $F^\oplus_{s, \region, a}(t_2) {-} F^\oplus_{s, \region, a}(t_1)$ is
  nonnegative.  Now by definition we have $F^\oplus_{s, \region, a}(t_2) {-} F^\oplus_{s, \region, a} (t_1) $ equals:
  \begin{eqnarray*}
    \lefteqn{t_2 {-} t_1
    + \mbox{$\sum\limits_{(C \!,\ell') \in \powC \times L}$}
    \delta[\ell,a](C,\ell') {\cdot} \big(
    F((\ell',\nu^{t_2}_C), (\ell', \region^C))
    {-} F((\ell',\nu^{t_1}_C), (\ell', \region^C))\big)} \\
    &=& t_2 {-} t_1 {-} \mbox{$\sum\limits_{(C \!,\ell') \in \powC \times L}$}
    \delta[\ell,a](C,\ell') {\cdot} \big(
    F((\ell',\nu^{t_1}_C), (\ell', \region^C))
    {-} F((\ell',\nu^{t_2}_C), (\ell', \region^C))\big)\\
    & \geq  & t_2 {-} t_1 {-}
    \mbox{$\sum\limits_{(C \!{,}\ell') \in \powC \times L}$}
    \delta[\ell{,}a](C{,} \ell') {\cdot} (t_2 {-} t_1) \\ & \geq &  0
  \end{eqnarray*}
  where the inequality is due to the fact the $F$ is monotonically
  decreasing and nonexpansive.  
  \qed
\end{proof}

\end{document}